\documentclass[11pt]{article}  
\usepackage{fullpage} 

\usepackage{amsthm}
\usepackage{amsmath}
\usepackage{amssymb}
\usepackage{todonotes}
\usepackage{mathtools}
\usepackage{xspace}
\usepackage{xcolor}
\usepackage{verbatim}
\usepackage{enumitem}
\usepackage{multirow}
\usepackage{nicefrac}
\usepackage{hyperref}
\usepackage{cite}
\usepackage{algorithm}
\usepackage{algpseudocode}

\usepackage{tikz}
\usetikzlibrary{positioning,decorations.pathreplacing}

\newcommand{\outdeg}{\operatorname{outdeg}}
\newcommand{\free}{\mathsf{free}}
\newcommand{\dem}{\mathsf{demand}}

\newcommand{\ind}{\textsc{Index}\xspace}

\newtheorem{theorem}{Theorem}
\newtheorem{lemma}{Lemma}

\theoremstyle{definition}
\newtheorem{definition}{Definition}

\title{One Color Makes All the Difference in  the Tractability of Partial Coloring in  Semi-Streaming}

\author{
Avinandan Das\thanks{Email: {\tt avinandan.das@aalto.fi} }  \\
{\small Department of Computer Science}\\
{\small Aalto University} \\
{\small Finland}
}

\date{}

\begin{document}
\maketitle

\begin{abstract}



This paper investigates the semi-streaming complexity of \textit{$k$-partial coloring}, a generalization of proper graph coloring. For $k \geq 1$, a $k$-partial coloring requires that each vertex $v$ in an $n$-node graph is assigned a color such that at least $\min\{k, \deg(v)\}$ of its neighbors are assigned colors different from its own. This framework naturally extends classical coloring problems: specifically, $k$-partial $(k+1)$-coloring and $k$-partial $k$-coloring generalize $(\Delta+1)$-proper coloring and $\Delta$-proper coloring, respectively.

	Prior works of Assadi, Chen, and Khanna [SODA~2019] and Assadi, Kumar, and Mittal [TheoretiCS~2023] show that both $(\Delta+1)$-proper coloring and $\Delta$-proper coloring admit one-pass randomized semi-streaming algorithms. We explore whether these efficiency gains extend to their partial coloring generalizations and reveal a sharp computational threshold : while $k$-partial $(k+1)$-coloring admits a one-pass randomized semi-streaming algorithm, the $k$-partial $k$-coloring  remains semi-streaming intractable, effectively demonstrating a ``dichotomy of one color'' in the streaming model.


\end{abstract}
\thispagestyle{empty}

\newpage
\setcounter{page}{1}
\section{Introduction}

\label{sec:intro}
\paragraph{Partial coloring.}
 A proper vertex coloring requires that every edge is bichromatic, i.e., every vertex $v$ has \emph{all}
its neighbors colored differently from $v$.
Partial coloring relaxes this requirement by only demanding \emph{enough} differently-colored neighbors.
Formally, for integers $k \ge 0$ and $c \ge 1$, a \emph{$k$-partial $c$-coloring} of a graph $G=(V,E)$ is a map
$\chi:V \to [c]$ such that for every vertex $v$,
\[
\bigl|\{u \in N_G(v) : \chi(u) \neq \chi(v)\}\bigr|
\;\ge\;
\min\{k,\deg_G(v)\}.
\]
Equivalently, if $\deg_G(v)\ge k$ then $v$ must have at least $k$ neighbors of a different color, whereas
if $\deg_G(v)<k$ then all neighbors of $v$ must have a different color. Specifically, $\Delta$-partial coloring implies a proper coloring of $G$ where $\Delta$ is the max. degree of $G$.

The most prominent parameter choices are $k$-partial $(k+1)$-coloring and $k$-partial $k$-coloring.
These notions generalize the classical degree-based benchmarks:
when $k=\Delta$, a $\Delta$-partial $(\Delta+1)$-coloring is exactly a proper $(\Delta+1)$-coloring,
and a $\Delta$-partial $\Delta$-coloring is exactly a proper $\Delta$-coloring.

A simple greedy argument in fact shows that a $k$-partial $(k+1)$-coloring always exists. Initialize all vertices to color~1 and consider each vertex sequentially. For each considered vertex $v$, let $C(v)$ be the set of colors present in the neighborhood of $v$. If $|C(v)|= k+1$, then the vertex retains it's color. Otherwise, it recolors itself with a color from the set $[k+1]\smallsetminus C(v)$. The resulting coloring is $k$-partial  because (1)~node~$v$ has at least $\min\{k,\deg(v)\}$ neighbors with a color different from its own color when it adopts its final color, and (2)~two neighboring vertices with different colors at some time~$t$ during the execution of the greedy algorithm will remain with different colors at any time $t'\geq t$.  

A central question in this paper is whether the semi-streaming tractability known for $(\Delta+1)$- and $\Delta$-proper coloring~\cite{ACK19,AKM23} extends to these partial-coloring generalizations as well.

\paragraph{A simple Streaming Algorithm.}

As a warm-up, a simple one-pass approach already yields a $k$-partial $(k+o(k))$-coloring.

A useful lens is to view partial coloring through a \emph{witness} subgraph. Suppose we find a spanning subgraph $G' \subseteq G$ such that every vertex $v$ retains at least $\min\{k,\deg_G(v)\}$ incident edges in $G'$. Then any \emph{proper} coloring of $G'$ with $c$ colors is automatically a $k$-partial $c$-coloring of $G$: each vertex has enough neighbors inside $G'$, and properness makes all witness edges bichromatic.

We run the \emph{filter} from Section~\ref{subsec:algorithm} to construct online a spanning subgraph $G_f=(V,E_f)$:
when an edge $\{u,v\}$ arrives, we accept it into $E_f$ iff the current degree of $u$ in $G_f$ is $<k$
or the current degree of $v$ in $G_f$ is $<k$. Lemma~\ref{lemma:witness-implies-partial} in Section~\ref{subsec:correctness} proves that $G_f$ is a valid witness, and by Lemma~\ref{lemma:degenerate}, that it  is $k$-degenerate.
Hence we may apply the degeneracy-based semi-streaming coloring algorithm~\cite{BCG20} to properly color $G_f$ with $k+o(k)$ many colors in semi-streaming.



\paragraph{\boldmath Challenges in getting down to $k+1$ colors.}

 The witness perspective also clarifies why achieving $k+1$ colors requires new ideas.

\begin{description}
	\item[1. Selecting the right streamed edges.] In semi-streaming we must decide online which edges to keep in the witness $G'$ under $\widetilde{O}(n)$ space.
The greedy filter above yields only a $k$-degenerate witness.
However, properly coloring an \emph{arbitrary} $k$-degenerate graph with $k+1$ colors is not feasible in one-pass semi-streaming space (refer to~\cite{BCG20}).
Thus, the first challenge is to maintain a witness that is not merely sparse, but structured enough that $(k+1)$-coloring becomes tractable.

\item[2.  Even with a good witness, what coloring approach applies?] Assume we somehow manage to keep a ``nice'' witness $G'$.
It is still not immediate what algorithmic approach should be deployed to color it with $k+1$ colors in one pass.
At first glance one might hope that the witness has $\Delta$ bounded by $k$ and hence that one can adapt
$(\Delta+1)$-coloring ideas.
However, this intuition can fail sharply: there exist graphs for which any $k$-partial $(k+1)$-coloring
essentially forces a proper coloring, while the graph still contains vertices of degree strictly larger than $k$;
see Fig.~\ref{fig} for such examples.
Therefore, one cannot rely on a maximum-degree bound, and the coloring algorithm must exploit a different
structural handle than naive $\Delta \le k$ reasoning.
\end{description}

\paragraph{\boldmath $k$-partial $k$-coloring.}

A tight classical characterization of when a graph admits a proper $\Delta$-coloring is given by Brooks' theorem (and its many variants and proofs); in particular, for a connected graph with maximum degree $\Delta$, a proper $\Delta$-coloring exists unless the graph is a clique (or an odd cycle when $\Delta=2$); see, e.g., the survey~\cite{CR15}. 
More recently, it was proven in~\cite{BDGJ25} that for every fixed constant $k\ge 3$, deciding whether a graph admits a $k$-partial $k$-coloring is NP-complete; their reduction replaces each edge of the input graph by a constant-size ``$k$-edge gadget'' (see Section~\ref{sec:lower-bound} and Fig.~\ref{fig}). 
At the same time, the problem of $k$-partial $k$-coloring for any constant $k$ is trivially semi-streaming tractable; implement the simple streaming algorithm discussed in Section~\ref{sec:intro} and store the filtered graph $G_f$ (which has degeneracy bounded by $k$ which is a constant and therefore can be stored offline in $\widetilde{O}(n)$ space) and find a coloring offline.  
This phenomenon is not unique: there are NP-hard problems that are nonetheless trivial  in semi-streaming space, e.g. deciding whether a cubic graph is $3$-edge-colorable is NP-complete~\cite{H81}, cubic graphs have only $3n/2$ edges and hence can be stored exactly in semi-streaming space.
Consequently, the question remains whether the one-pass semi-streaming tractability known for proper $\Delta$-coloring~\cite{AKM23} extends to $k$-partial $k$-coloring as well.

\section{Related Work}
\label{sec:related-work}

Graph problems in streaming are commonly studied in the \emph{semi-streaming} model introduced in~\cite{FKMSZ05}, where the graph arrives as an edge stream and the algorithm uses
$\tilde{O}(n)$ space. For background on techniques and representative results in graph streams, see survey~\cite{M14}.

A central tool behind modern sublinear graph-coloring algorithms is the \emph{palette sparsification theorem}~\cite{ACK19}, which shows that sampling $O(\log n)$ colors per vertex from $\{1,\dots,\Delta+1\}$
typically preserves the existence of a proper $(\Delta+1)$-coloring, enabling single-pass semi-streaming
algorithms in dynamic streams. Subsequent work studied palette sparsification in special classes of graphs, like in triangle free graphs with a smaller palette of colors as well as generalizing it to $(\deg+1)$-coloring~\cite{AA20}. It was further generalized to $(\deg+1)$-list coloring in~\cite{HKNT22}. Recently a substantially simpler (albeit slightly weaker) proof of palette sparsification has appeared in~\cite{AY25}.

Proper $\Delta$-coloring was also studied in the streaming framework and was proven to be semi-streaming tractable in~\cite{AKM23}. Coloring was also studied w.r.t \emph{degeneracy} of a graph. In~\cite{BCG20}, a semi-streaming algorithm was designed which colors a graph with degeneracy $\kappa$ in $\kappa+o(\kappa)$ many colors. They also proved that $\kappa+1$ coloring of a $n$-vertex graph requires $\Omega(n^2)$ bits of space. We adapt their lower bound proof to the partial coloring setting in Section~\ref{sec:lower-bound}.

The notion of \emph{partial} (or \emph{$k$-partially proper}) coloring has appeared previously in distributed
graph algorithms before~\cite{NS93, BHLOS19, DFR23, BDGJ25}. Specifically,~\cite{BDGJ25} studied the problem of $k$-partial $k$-coloring and proved that the problem is NP-complete and proved that the problem is``global" for every constant $k$.  

Our work complements these results by establishing an analogous \emph{one-color} threshold phenomenon in the
semi-streaming model.

\section{Our Results} 

The results of this paper demonstrate  a sharp \emph{one-color dichotomy} in the semi-streaming setting: allowing \emph{$k{+}1$} colors makes $k$-partial coloring tractable in semi-streaming space, whereas
restricting to exactly \emph{$k$} colors can force super-linear memory.

\begin{theorem}
There exists a one-pass randomized semi-streaming algorithm that, with high probability, produces a $k$-partial coloring of an input graph $G$ using $k+1$ colors,  
where $k \in \mathbb{N}$ is given as input prior to the stream and $G$ is presented as an insertion-only stream.
\end{theorem}

Section~\ref{sec:semi-streaming-algorithm} is dedicated to proving the theorem. In section~\ref{subsec:palette-sparsification}, we  begin by identifying a sparse spanning subgraph $G'\subseteq G$ that serves as a \emph{witness} for partial coloring: a proper $(k+1)$-coloring of $G'$ immediately induces a $k$-partial $(k+1)$-coloring of $G$. We then  prove a palette sparsification theorem for proper $(k+1)$-coloring of the witness graph $G'$ and finally in Section~\ref{subsec:algorithm}, we adapt the sparsification lemma to obtain a one-pass semi-streaming algorithm for the coloring.

In Section~\ref{sec:lower-bound}, we augment this upper bound by a fine-grained  super linear  lower bound for $k$-partial $k$-coloring. More specifically, we prove the following theorem.

\begin{theorem}
    For every constant $\varepsilon \geq 0$, let $G$ be an $n$-vertex graph presented as an adversarial insertion-only stream, and let the partial coloring parameter be $k = \Theta(n^{\frac{1}{3+\varepsilon}})$.
    Any one-pass randomized streaming algorithm $\mathcal{A}$ that solves $k$-partial $k$-coloring on $G$ with probability at least $2/3$ requires $\Omega\left(n^{1 + \frac{1}{3+\varepsilon}}\right)$
    bits of memory. In particular, for $\varepsilon=0$, the algorithm requires $\Omega(n^{4/3})$ bits.
\end{theorem}

Without loss of generality, we assume that $k= \omega(\text{poly}\log n)$. Otherwise, the problem becomes trivially semi-streaming tractable; run the filter of the simple algorithm in Sec~\ref{sec:intro} and store the edges which pass the filter. The resultant graph will have degeneracy bounded by $k$ and therefore, can be stored and colored offline in semi-streaming space.


\subsection{Technical Overview}
\paragraph*{Semi-Streaming Algorithm.}

Our approach is based on the observation that partial coloring can be viewed as a proper coloring
problem on a suitable spanning subgraph.
By definition, any $k$-partial coloring of a graph $G$ implicitly induces a proper coloring on
a subgraph $G' \subseteq G$, which acts as a \emph{witness} for partial coloring:
if a coloring is proper on $G'$, then every vertex has at least
$\min\{k,\deg_G(v)\}$ neighbors colored differently in $G$.

We exploit this perspective by explicitly constructing such a witness subgraph.
Starting from the input graph $G$, we repeatedly delete edges whose both endpoints currently
have degree strictly larger than $k$.
This process terminates with a subgraph $G'$ in which every edge has at least one endpoint
of degree at most $k$.
A key property of this construction is that any proper coloring of $G'$ is automatically
a valid $k$-partial coloring of the original graph $G$.

Once the problem is reduced to properly coloring the witness graph $G'$, we aim to apply
\emph{palette sparsification}, following the general strategy of~\cite{ACK19}.
However, unlike the standard $(\Delta+1)$-coloring setting, palette sparsification does not
immediately apply here: vertices sample colors from the fixed palette $[k+1]$, rather than
from degree-dependent palettes. Moreover, we might end up in situations where the graph $G'$ has to be properly colored but $k<\Delta(G')$ (refer to Fig.~\ref{fig}).
This mismatch prevents a direct application of existing sparsification results.

To overcome this difficulty, we adapt palette sparsification to the structure of the witness graph.
The key idea is to exploit the degree asymmetry enforced by the witness construction.
We partition the vertices of $G'$ into \emph{high-degree} and \emph{low-degree} vertices.
By construction, the set of high-degree vertices (those with degree greater than $k$)
forms an independent set.
This allows us to color all high-degree vertices first using a small subset of colors.
In particular, we show that they can be colored with at most $k/4$ colors with high probability,
using only their randomly sampled palettes.

After fixing the colors of the high-degree vertices, we proceed to color the low-degree vertices.
Each low-degree vertex has at most $k$ neighbors in total, and only a small number of colors
are forbidden by already-colored neighbors.
Although the palettes of these vertices were sampled from $[k+1]$, we show that the
\emph{effective palettes} obtained after removing forbidden colors remain sufficiently large
and retain the necessary randomness.
This enables the application of a variant of  palette sparsification theorem for $(\deg+1)$-list coloring~\cite[Theorem~4]{HKNT22}
to complete the coloring.

To translate this existential argument into a streaming algorithm, we introduce a three-phase
procedure.
In \emph{Phase~1 (Filtering)}, edges are discarded once both endpoints have accumulated degree $k$,
ensuring that the maintained graph remains a witness.
In \emph{Phase~2 (Palette Sparsification)}, edges whose endpoints have disjoint sampled palettes
are discarded, since they cannot induce color conflicts; only potentially constraining edges
are stored.
Finally, in \emph{Phase~3 (Offline Solving)}, the algorithm solves the resulting
demand-partial list coloring instance offline.
By the adapted palette sparsification argument, a solution exists with high probability,
and any such solution corresponds to a proper coloring of $G'$ and hence to a valid
$k$-partial coloring of $G$.

\paragraph*{Lower Bound.}

We base ourselves on the lower bound technique of~\cite{BCG20} who prove that there exists a $\kappa$ such that coloring an input $n$-vertex $\kappa$-degenerate graph $G$ with $\kappa+1$ colors requires $\Omega(n^2)$ bits in insertion-only stream. Their proof follows a standard reduction from the communication problem \ind where Alice is given a vector $X\in \{0,1\}^N$ and Bob, an index $i\in [N]$ and Bob outputs the value $X_i$.  They follow a standard reduction procedure of shaping $X$ as a square matrix $A[n\times n]$ and the index is mapped as $(g,h)$. These are then encoded as a graph $G'$ where two sets of vertices $U$ and $W$ represent the rows and columns of $A$ such that for vertices $u_i\in U$ and $w_j\in W$, $\{u_i,w_j\}\in E(G')$ if and only if $A[i,j] = 1$. The graph $G'$ is encoded in such a way that \emph{for any $\kappa+1$ coloring $\gamma$, for vertices $u_g\in U$ and $w_h\in W$, $\gamma(u_g)\neq \gamma(w_h)$ if and only if $A[g,h]=1$}.

We follow the same strategy and reshape the vector $X$ into a rectengular matrix $A[(k-1)\times \ell]$ and the index $i$ mapped as $(t,q)$ and Alice and Bob construct a graph $G$ encoding $A$ and  $(t,q)$. Let $R$ and $C$ be the  row and column vertices in $G$ corresponding to $A$. Alice  ensures that for each vertex $v\in C$ : 1) $\deg_G(v) = k-1$ and 2) $\chi(N_G[v])=[k]$ and for each vertex $u\in R$, $|\chi(N_G[u])| =k-1$. Bob, based on the index $(t,q)$ adds edges to the graph ensuring that $\chi(N_G[r_t]) = [k-1]$.   This is ensured via the special gadgets \emph{edge-gadgets} and \emph{color repeaters} which, when added between two vertices force the two vertices to take different and same colors respectively. These properties, along with the fact that any $k$-partial $k$-coloring $\chi$ will color all the incident edges of $v$ properly w.r.t it ensures that  $G$ has to satisfy the property that for   vertices $r_t\in R$ and $c_q\in C$, $\chi(r_t)\neq \chi(c_q)$ if and only if $A[t,q]=1$. 

Any streaming algorithm acts as a one-way communication protocol and in the end, the coloring of the vertices reveal the bit $A[t,q]$. The space bound follows from the fact that the $|V(G)| = O(k +\ell)$, setting $k^{3+\varepsilon} = O(\ell) $ and that \ind requires $\Omega(k\ell)$ bit one-way communication bandwidth. 

\section{Preliminaries}\label{section:prelims}

We write $[k] := \{1,2,\dots,k\}$. All logarithms are base two unless stated otherwise.
All graphs considered in this paper are simple, undirected, and have $n$ vertices.

Throughout the paper, we sample colors \emph{independently and uniformly at random with
replacement} from a given palette. Formally, sampling with replacement from a finite set $U$ means that each draw is an
independent random variable distributed uniformly over $U$. Consequently, the same
color may appear multiple times in a sampled list. 

\begin{lemma}\label{lemma:conditional_sampling}
	Let $X$ be a color sampled  uniformly at random from $[k+1]$. Let $U_u\subseteq [k+1]$ be any fixed subset. Conditioned on the event that $X\in U_u$, the color $X$ is uniformly distributed over $U_u$.
\end{lemma}

\begin{proof}
For any $c\in U_u$,
\[
\Pr[X=c \mid X\in U_u]
=
\frac{\Pr[X=c]}{\Pr[X\in U_u]}
=
\frac{1/(k+1)}{|U_u|/(k+1)}
=
\frac{1}{|U_u|}.
\]
\end{proof}

We use standard Chernoff bounds for sums of independent Bernoulli random variables~\cite{MU17}.

Let $X=\sum_{i=1}^m X_i$, where the $X_i$ are independent and take values in $\{0,1\}$.
Then, for any $0<\delta<1$,
\[
\Pr\bigl[X \le (1-\delta)\mathbb{E}[X]\bigr]
\;\le\;
\exp\!\left(-\frac{\delta^2}{2}\mathbb{E}[X]\right).
\]
\paragraph*{Degeneracy.}

A graph $G=(V,E)$ is \emph{$d$-degenerate} if every non-empty subgraph $H\subseteq G$
contains a vertex of degree at most $d$.

Equivalently, $G$ is $d$-degenerate if there exists an ordering of the vertices
$v_1,\dots,v_n$ such that for every $i$, vertex $v_i$ has at most $d$ neighbors among
$\{v_{i+1},\dots,v_n\}$. We orient each edge $\{v_i,v_j\}$ with $i<j$
from $v_i$ to $v_j$. The \emph{outdegree} of $v_i$ is defined as
\[
\outdeg(v_i)
=
\bigl|\{\, v_j : j>i \text{ and } \{v_i,v_j\}\in E \,\}\bigr|.
\]
With this orientation, $G$ is $d$-degenerate if and only if there exists an ordering such
that $\outdeg(v_i)\le d$ for all $i\in[n]$.

\paragraph*{Semi-Streaming Algorithms.}

We work in the \emph{semi-streaming model} for graph algorithms.

An input graph $G=(V,E)$ is presented as an insertion-only stream of edges. The algorithm knows the vertex set $V$ apriori and
processes the stream of edges in one pass while maintaining a memory state of size at most
$\widetilde{O}(n)$ bits, where polylogarithmic factors in $n$ are suppressed.

The algorithm may be randomized and succeeds with high probability (i.e. with probability at least $\left(1 - 1/n^c\right)$ for some constant $c$). At the end of the stream, the algorithm outputs a solution to the
given graph problem.

\paragraph*{Communication Complexity.}

We use communication complexity to prove lower bounds for streaming algorithms.

In the \emph{one-way communication model}, there are two players, Alice and Bob. Alice
receives an input $x$, Bob receives an input $y$, and Alice sends a single message to
Bob. Based on this message and his input, Bob must compute a function $f(x,y)$.

The communication cost is the number of bits sent from Alice to Bob. Protocols may be
randomized and are required to succeed with probability at least $2/3$.

A standard connection between streaming algorithms and communication complexity is that
any one-pass streaming algorithm using $S$ bits of space yields a one-way communication
protocol with communication cost $S$, where Alice simulates the algorithm on the first
part of the stream and sends the memory state to Bob.

\paragraph*{Index.} The \ind problem  is defined as follows:
Alice holds a bit string $X\in\{0,1\}^N$, Bob holds an index $j\in[N]$, and Bob must output
the bit $X_j$.

It is well known that any one-way randomized protocol for Index that succeeds with
probability at least $2/3$ requires $\Omega(N)$ bits of communication~\cite{JKS08}.

%
\section{Semi-Streaming Algorithm for $k$-partial $(k+1)$-coloring}\label{sec:semi-streaming-algorithm}

\begin{figure}[t]
   \centering
    \includegraphics[width=0.4\linewidth]{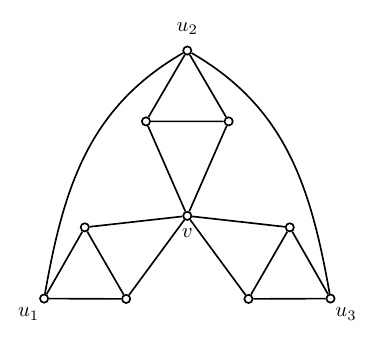}
    \caption{Example of a graph where $3$-partial coloring implies a proper coloring. Observe that the vertices $v$ and $u_2$ have degree greater than $3$ but are surrounded by degree $3$ vertices which enforce a proper coloring on them.}
  \label{fig}
\end{figure}

This section is dedicated to design, correctness and analysis of a single pass insertion-only semi-streaming algorithm for partial coloring. Specifically, we prove the following theorem. 

\begin{theorem}\label{thm:k-partial-streaming}
There exists a one-pass randomized semi-streaming algorithm that, with high probability, produces a $k$-partial coloring of an input graph $G$ using $k+1$ colors,  
where $k \in \mathbb{N}$ is given as input prior to the stream and $G$ is presented as an insertion-only stream.
\end{theorem}

\subsection{Palette Sparsification for Partial Coloring}\label{subsec:palette-sparsification}


\begin{definition}[Witness for Partial Coloring]
    Let $G=(V,E)$ be a graph and $k$ a positive integer. A \emph{witness for $k$-partial coloring} is a subgraph $G' \subseteq G$ obtained by the following procedure:
    \begin{enumerate}
        \item Initialize $G' = G$.
        \item While there exists an edge $\{u,v\} \in E(G')$ such that $\deg_{G'}(u) > k$ and $\deg_{G'}(v) > k$, remove $\{u,v\}$ from $E(G')$.
    \end{enumerate}
    The resultant graph satisfies that for every edge $\{u,v\} \in E(G')$, $\min\{\deg_{G'}(u), \deg_{G'}(v)\} \leq k$.
\end{definition}

\begin{lemma}\label{lemma:witness_graph}
    Let $G'$ be a witness for $k$-partial coloring of $G$. If $\chi$ is a proper coloring of $G'$, then $\chi$ is a $k$-partial coloring of $G$.
\end{lemma}

\begin{proof}
    We first establish a lower bound on the degrees in $G'$. We claim that for every vertex $v \in V(G)$, 
    \[
        \deg_{G'}(v) \geq \min\{k, \deg_G(v)\}.
    \]
    Consider the edge deletion procedure defining $G'$. An edge incident to $v$ is removed only if $\deg_{current}(v) > k$. Consequently, the degree of $v$ can never be reduced to a value less than $k$, as the removal condition would fail before such a reduction could occur.
    \begin{itemize}
        \item If $\deg_G(v) \leq k$, the condition $\deg(v) > k$ is never met, so no edges incident to $v$ are ever removed. Thus, $\deg_{G'}(v) = \deg_G(v)$.
        \item If $\deg_G(v) > k$, edges incident to $v$ may be removed, but the process must halt for $v$ if its degree reaches $k$. Thus, $\deg_{G'}(v) \geq k$.
    \end{itemize}
    Combining these cases, we have $\deg_{G'}(v) \geq \min\{k, \deg_G(v)\}$.

    Now, since $\chi$ is a proper coloring of $G'$, $v$ is assigned a color distinct from all its neighbors in $G'$. The number of such neighbors is $\deg_{G'}(v)$.
    Using our lower bound, the number of neighbors of $v$ in $G$ colored differently than $v$ is at least $\min\{k, \deg_G(v)\}$. This satisfies the definition of a $k$-partial coloring.
\end{proof}

Our aim now is to prove palette sparsification for the witness graph $G'$. Before proceeding to prove the theorem, we rephrase the palette sparsification theorem for $(\deg+1)$-list coloring~\cite[Theorem~4]{HKNT22}  which we will use as a blackbox in our proof. 

\begin{theorem}[Palette sparsification for $(\deg+1)$-list coloring]\label{thm:palette-sparsification}
	Let $G$ be an $n$-vertex graph. Suppose each vertex $v \in V(G)$ is assigned a list $L(v)$ of $\deg(v)+1$ colors. For each vertex $v$, independently sample a subset $L'(v) \subseteq L(v)$ of size $\Theta(\log^2 n)$ uniformly at random. Then, with high probability, there exists a proper coloring
\[
	\mathcal{C} \colon V(G) \to \bigcup_{v\in V(G)} L(v)
\]
such that $\mathcal{C}(v) \in L'(v)$ for every vertex $v \in V(G)$.
\end{theorem}

 Without loss of generality, for Theorem~\ref{thm:palette-sparsification}, we assume that for each vertex $v\in V$,  $|L'(v)| =  \alpha\log^2 n$ for some fixed constant $\alpha $ and each color in $L'(v)$ is independently sampled uniformly at random with replacement from the palette of $L(v)$.

\begin{theorem}\label{theorem:palette-sparsification-partial-coloring}
	Given an $n$-vertex graph $G$ and a positive integer $k$, let $G'$ be the witness for $k$-partial coloring of $G$.  Suppose that for every vertex $v\in V(G')$, we sample a list $L(v)$ of $\Theta(\log^2 n)$ colors uniformly at random from the palette $[k+1]$. Then, with high probability, there exists a proper $(k+1)$-coloring $\mathcal{C}' \colon V(G')\to [k+1]$ of $G'$ such that $\mathcal{C}'(v)\in L(v)$ for all $v\in V(G')$.
\end{theorem}

\begin{proof}
    We prove the existence of such a coloring using a two-phase randomized procedure.
    First, for every vertex $v \in V(G')$, we split the sampled list $L(v)$ into two disjoint sub-lists $L_1(v)$ and $L_2(v)$ such that $|L_1(v)| = C_1 \log n$ and $|L_2(v)| = C_2 \log^2 n$ for sufficiently large constants $C_1, C_2$.

    We partition the vertex set $V(G')$ into two sets: $S = \{v \in V(G') \mid \deg_{G'}(v) > k\}$ and $R = V(G') \setminus S$.
    By the definition of a witness (specifically the constraint that for every edge $\{u,v\}$, $\min\{\deg_{G'}(u), \deg_{G'}(v)\} \leq k$), no two vertices with degree strictly greater than $k$ can be adjacent. Consequently, $S$ is an independent set in $G'$.

    The coloring proceeds in two phases:
    \begin{itemize}
        \item \textbf{Phase 1 (Coloring $S$):} We assign colors to vertices in $S$ using only the palettes $L_1(v)$.
        \item \textbf{Phase 2 (Coloring $R$):} We extend the coloring to $R$ using the palettes $L_2(v)$.
    \end{itemize}

    \paragraph{Phase 1: Coloring the Independent Set $S$.}
    We claim that with high probability, $S$ can be properly colored using a subset of colors $T \subseteq [k+1]$ of size at most $k/4$, such that for every $v \in S$, the assigned color belongs to $L_1(v) \cap T$.

    To show this, we construct a bipartite \emph{palette graph} $P = (S \cup [k+1], E_P)$, where an edge $\{v, c\}$ exists if and only if $c \in L_1(v)$. A subset of colors $T \subseteq [k+1]$ is said to \emph{dominate} $S$ in $P$ if every vertex $v \in S$ has at least one neighbor in $T$. If such a set $T$ exists, we can define a valid coloring $\chi_1$ for $S$ by setting $\chi_1(v) = \min \{c \in T \cap L_1(v)\}$. Since $S$ is an independent set, no conflicts can arise between vertices in $S$ regardless of the color choices.

    We now show that there exists $T \subseteq [k+1]$ of size $k/4$ that dominates $S$ with high probability.
    Let $T$ be an arbitrarily chosen but fixed subset of $[k+1]$ with $|T| = k/4$. For any fixed vertex $v \in S$, the probability that $L_1(v) \cap T = \emptyset$ is
    \[
        \mathbb{P}[L_1(v) \cap T = \emptyset] \leq \left(1 - \frac{|T|}{k+1}\right)^{|L_1(v)|} \approx \left(1 - \frac{1}{4}\right)^{C_1 \log n} \leq e^{-0.25 C_1 \log n} = n^{-0.25 C_1}.
    \]
    By choosing $C_1$ sufficiently large (e.g., $C_1 \geq 12$), we ensure this probability is at most $n^{-3}$.
    Taking a union bound over all vertices in $S$ (where $|S| \leq n$), the probability that there exists any vertex in $S$ not dominated by $T$ is at most $n \cdot n^{-3} = n^{-2}$.
    Thus, with high probability, there exists a dominating set $T$ of size $k/4$, and consequently a valid coloring $\chi_1$ for $S$ using only colors from $T$.

    \paragraph{Phase 2: Extending the coloring to $R = V \setminus S$.}
Let $\chi_1$ be the coloring of $S$ obtained in Phase 1. Recall that Phase 1 ensures the set of colors used in $S$, denoted $C(S)$, satisfies $|C(S)| \leq k/4$.

We extend this coloring to $R$. For each vertex $u \in R$, let $C_{\text{forb}}(u) = \{\chi_1(v) \mid v \in N_{G'}(u) \cap S\}$ be the set of colors forbidden by neighbors in $S$. We define the target palette for $u$ as $U_u = [k+1] \setminus C_{\text{forb}}(u)$.

Since $u \in R$, we know $\deg_{G'}(u) \le k$. The number of forbidden colors from $S$ is at most the number of neighbors $u$ has in $S$, which is $\deg_{G'}(u) - \deg_R(u)$. Therefore, it follows easily that $|U_u| \geq \deg_R(u) + 1$ meaning  that a valid extension exists if we had access to $U_u$ and coloring $R$ reduces to solving $(\deg+1)$-list coloring. 


	At this point, we cannot directly apply  Theorem~\ref{thm:palette-sparsification} to $R$ as it requires that for vertex $u$, the random list of size at least $\alpha\log^2 n$ be sampled uniformly from the available palette $U_u$ for some constant $\alpha$ while $L_2(u)$ is sampled from $[k+1]$. Let $L_{\text{eff}}(u) = L_2(u) \cap U_u$. We claim that $L_{\text{eff}}(u)$ satisfies the that 
	\begin{enumerate}
		\item Each color in  $L_{\text{eff}}(u)$ is uniformly distributed over $U_u$.
		\item $|L_{\text{eff}}(u)|\geq \alpha\log^2 n$ with high probability .
	\end{enumerate}


	The first item follows immediately from Lemma~\ref{lemma:conditional_sampling}.

	
Now, we show the second item. Since $C_{\text{forb}}(u) \subseteq C(S)$ and $|C(S)| \leq k/4$, we have $|U_u| \geq \frac{3}{4}(k+1)$.
	The size $|L_{\text{eff}}(u)|$ follows a Binomial distribution $B(|L_2(u)|, p)$ with $p = \frac{|U_u|}{k+1} \geq 3/4$. Fix $C_2 = 8\alpha/3$. 
	Using a Chernoff bound, with $|L_2(u)| = (8/3)\alpha\log^2 n$ and $\mathbb{E}[|L_{\text{eff}}(u)|]\geq 2\alpha\log^2 n$, the effective list size $|L_{\text{eff}}(u)|$ is at least $\alpha\log^2 n$ with high probability. Applying union bound, this holds to every vertex in $R$ with high probability.
	
Since for every vertex $w\in R$, $|L_{\text{eff}}(w)|$ is sufficiently large and conditionally uniform over available palette $U_w$, we can invoke the Palette Sparsification Theorem on the induced subgraph $G'[R]$ with palettes $U_w$. This guarantees that a valid coloring $\chi_2$ consistent with the lists exists with high probability.

We now combine the results of both phases to bound the total failure probability.
Let $E_1$ denote the event that Phase 1 successfully colors $S$ with a subset of colors $T$ such that $|T| \le k/4$. We established that $\mathbb{P}[E_1] \geq 1 - n^{-2}$.

Conditioned on the occurrence of $E_1$, the coloring $\chi_1$ and the forbidden sets $C_{\text{forb}}(u)$ are well-defined. Let $E_{2 \mid 1}$ denote the event that Phase 2 successfully extends this specific coloring to $R$ using the lists $L_2$. Our analysis in Phase 2 demonstrates that for any fixed successful outcome of Phase 1,
\[
	\mathbb{P}[E_{2 \mid 1}] \geq 1 - n^{-c} \text{ for some positive integer }c. 
\]
The event that a valid $(k+1)$-coloring exists for the entire graph $G'$ corresponds to the joint occurrence of both phases succeeding. By the chain rule of probability:
\[
	\mathbb{P}[\text{Success}] = \mathbb{P}[E_1] \cdot \mathbb{P}[E_{2 \mid 1}] \geq (1 - n^{-2})(1 - n^{-c}) \geq 1 - n^{-O(1)}.
\]
Thus, with high probability, there exists a proper $(k+1)$-coloring $\mathcal{C}'$ of $G'$ such that $\mathcal{C}'(v) \in L(v)$ for all $v \in V(G')$.
\end{proof}

\begin{algorithm}[t]
\caption{One-pass semi-streaming algorithm for $k$-partial $(k+1)$-coloring}
\label{alg:streaming-partial-coloring}
\begin{algorithmic}[1]

\State \textbf{Input:}\  Graph $G=(V,E)$ given as an insertion-only stream, integer $k$
\State \textbf{Output:}\ A $k$-partial $(k+1)$-coloring of $G$ (with high probability)

\State Let $s := C\log^2 n$ for a sufficiently large constant $C$.  For each $v\in V$, sample a list $L(v)$ of $s$ colors independently and uniformly at random from $[k+1]$ (with replacement)
	\State Initialize counters $\deg(v)\gets 0$ and $\free(v)\gets 0$ for all $v\in V$
\State Initialize an empty graph $H$

\For{each streamed edge $\{u,v\}\in E$}
    \If{$\deg(u) < k$ \textbf{or} $\deg(v) < k$}
        \State $\deg(u) \gets \deg(u) + 1$
        \State $\deg(v) \gets \deg(v) + 1$
        \If{$L(u)\cap L(v)=\emptyset$}
	    \State $\free(u)\gets \free(u)+1$
            \State $\free(v)\gets \free(v)+1$
        \Else
            \State Store the edge $\{u,v\}$ in $H$
        \EndIf
    \EndIf
\EndFor

\For{each $v\in V$}
    \State $\dem(v) \gets \max\bigl\{0,\min\{k,\deg(v)\}-\free(v)\bigr\}$
\EndFor

	\State \textbf{Offline:} find (by brute force) a coloring $\chi: V\rightarrow [k+1]$ which is a solution for the demand-partial list coloring instance $(H,\dem)$ such that $\chi(v)\in L(v)$ for all $v\in V$.

\State \Return $\chi$
\end{algorithmic}
\end{algorithm}


\subsection{Algorithm}\label{sec:algorithm}\label{subsec:algorithm}

We now describe our one-pass semi-streaming algorithm for computing a
	$k$-partial $(k+1)$-coloring. Before describing the algorithm, we first define \emph{demand-partial list coloring}.

	\begin{definition}[Demand-partial list coloring.]
Let $H=(V,F)$ be a graph and let each vertex $v\in V$ have a list $L(v)\subseteq [k+1]$. Let $\dem:V\to [k+1]$ be a demand function.
A coloring $\chi:V\to [k+1]$ is a \emph{solution to the demand-partial list coloring instance $(H,\dem)$}
if
\begin{enumerate}
	\item $\chi(v)\in L(v)$ for every $v\in V$, and
	\item  for every vertex $v\in V$, $\bigl|\{\,x\in N_H(v)\ :\ \chi(x)\neq \chi(v)\,\}\bigr|\ \ge\ \dem(v)$.

\end{enumerate}
In other words, $\dem(v)$ specifies how many neighbors of $v$ in $H$ must end up with a color
different from $\chi(v)$; we do not require $\chi$ to be a proper coloring of $H$.
\end{definition}

\subsubsection{Description of the Algorithm}
The algorithm knows the parameter $k$ apriori and the input is an insertion-only stream of edges of a $n$-vertex graph  $G=(V,E)$.
Let $s:=C\log^2 n$ for a sufficiently large constant $C$.
Before the stream starts, each vertex $v\in V$ independently samples a list $L(v)$ of $s$ colors
from $[k+1]$ uniformly with replacement.
For each vertex $v\in V$, the algorithm maintains two counters $\deg(v)$ and $\free(v)$ both initialized at $0$. 


	The algorithm has three phases. As an edge $\{u,v\}$ streams in, the algorithm processes it in two phases. The third phase is  post-processing  where the final coloring of the graph is computed.

	\begin{description}
		\item[Phase 1 : Filtering.] if $\deg(u)<k$ or $\deg(v)<k$, the edge $\{u,v\}$ proceeds to the next phase and the degrees are updated as follows :
			\[
				\deg(u)\gets \deg(u)+1,
				\qquad
				\deg(v)\gets \deg(v)+1.
			\]
			If $\deg(u)\ge k$ and $\deg(v)\ge k$, the edge is ignored and no state changes.

		\item[Phase 2 : Sparsification.] Conditioned on passing the filter, the algorithm checks whether the sampled lists are disjoint.
If $L(u)\cap L(v)=\emptyset$, then regardless of the eventual choices $\chi(u)\in L(u)$ and
$\chi(v)\in L(v)$, we necessarily have $\chi(u)\neq \chi(v)$.
Hence this edge does not need to be stored; instead we increment
\[
	\free(u)\gets \free(u)+1,
	\qquad
	\free(v)\gets \free(v)+1.
\]
Otherwise, if $L(u)\cap L(v)\neq \emptyset$, the edge may constrain the final solution and the
algorithm stores $\{u,v\}$.

		\item[Phase 3 : Post-Processing.] Let the graph stored after the end of the stream be $H = (V,F)$. For each vertex $v\in V$, the algorithm defines the demand as 
			\[
					\dem(v)\ \coloneqq\ \max\bigl\{0,\ \min\{k,\deg(v)\}-\free(v)\bigr\}.
			\]
			The algorithm computes offline the coloring $\chi : V\rightarrow [k+1]$ for the demand-partial list coloring $(H,\dem)$ where for each vertex $v\in V$, $\chi(v)\in L(v)$.
	\end{description}

\paragraph{\color{red}Remark.} The algorithm works in the insertion-only framework. The key challenge in adapting it to dynamic setup is the implementation of the filter in Phase~1. 
\subsubsection{Correctness}\label{subsec:correctness}
\begin{lemma}\label{lemma:witness-implies-partial}
Let $G'$ be the graph consisting of the edges that pass the filter of Algorithm~\ref{alg:streaming-partial-coloring}.
If $\chi$ is a $k$-partial coloring of $G'$,
then $\chi$ is a $k$-partial coloring of $G$.
\end{lemma}

\begin{proof}
Fix an arbitrary vertex $v\in V$. There are two cases.

	\begin{description}
		\item[Case 1 : $\deg_G(v)<k$.] Each edge that arrives and is incident on it passes through the filter. Hence, $N_G(v) = N_{G'}(v)$ and $\chi$ colors both neighborhoods properly w.r.t $v$ and trivially satisfies the condition of valid $k$-partial coloring on $v$.
		\item[Case 2 : $\deg_G(v)\geq k$.] The first $k$ incident edges in the stream pass through the filter. Hence, $\deg_{G'}(v)\geq k$. Given that $N_{G'}(v)\subseteq N_{G}(v)$ and $\chi$ is a $k$-partial coloring of $G'$, at least $k$ vertices in $N_G(v)$ are also colored differently than $v$ in $G$ satisfying the condition of $k$-partial coloring.   
	\end{description}
\end{proof}

\begin{lemma}\label{lem:kpartial-implies-demand}
	If $\chi$ is a  $k$-partial coloring of $G'$ such that $\chi(x)\in L(x)$ for all $x\in V(G')$, then for every vertex $v\in V(G')$,
\[
\bigl|\{u\in N_H(v):\chi(u)\neq\chi(v)\}\bigr|\ \ge\ \dem(v).
\]
\end{lemma}

\begin{proof}

	By the description of the algorithm, at the end of the filtering phase, for each vertex $v\in V(G')$, $\deg_{G'}(v) = \deg(v)$.

	Fix a vertex $v$. Split its neighborhood in $G'$ into stored and discarded parts:
\[
N_{G'}(v)=N_H(v)\ \cup\ \bigl(N_{G'}(v)\setminus N_H(v)\bigr).
\]
For every $u\in N_{G'}(v)\setminus N_H(v)$ we have $L(u)\cap L(v)=\emptyset$ by definition of $H$, and hence $\chi(u)\neq\chi(v)$.
Therefore,
\[
\bigl|\{u\in N_{G'}(v):\chi(u)\neq\chi(v)\}\bigr|
=
\free(v)\ +\ \bigl|\{u\in N_H(v):\chi(u)\neq\chi(v)\}\bigr|.
\]
Since $\chi$ is $k$-partial on $G'$, the left-hand side is at least $\min\{k,\deg_{G'}(v)\}$.
Rearranging yields
	\begin{align*}
\bigl|\{u\in N_H(v):\chi(u)\neq\chi(v)\}\bigr|
		&\ge \min\{k,\deg_{G'}(v)\}-\free(v)\\
		&\geq \max\{0,\min\{k,\deg_{G'}(v)\}-\free(v)\}
		= \dem(v).
	\end{align*}

\end{proof}

\begin{lemma}\label{lem:existence-demand-solution}
With high probability over the sampled lists $\{L(v)\}_{v\in V}$,
there exists a coloring $\chi:V\to[k+1]$ such that
\begin{itemize}
    \item $\chi(v)\in L(v)$ for all $v\in V$, and
    \item $\bigl|\{u\in N_H(v):\chi(u)\neq\chi(v)\}\bigr|\ge \dem(v)$ for all $v\in V$.
\end{itemize}
\end{lemma}

\begin{proof}
Let $W$ be any witness for $k$-partial coloring of $G'$. By Theorem~\ref{theorem:palette-sparsification-partial-coloring} applied to $W$,
with high probability there exists a proper $(k+1)$-coloring $\chi$ of $W$ such that $\chi(v)\in L(v)$ for all $v\in V$.

By Lemma~\ref{lemma:witness-implies-partial}, $\chi$ is a $k$-partial coloring of $G'$ and is list respecting.
Finally, Lemma~\ref{lem:kpartial-implies-demand} implies that $\chi$ satisfies the
demand constraints on $H$.
\end{proof}
\begin{theorem}\label{thm:algorithm-correctness}
Algorithm~\ref{alg:streaming-partial-coloring} outputs a $k$-partial $(k+1)$-coloring of $G$
with high probability.
\end{theorem}
\begin{proof}
By Lemma~\ref{lem:existence-demand-solution}, with high probability there exists at least one list-respecting coloring $\chi$ that satisfies the demand constraints on $H$. The offline brute-force step finds such a coloring whenever it exists.

Assume the brute-force step outputs such a $\chi$. Then, by construction, $\chi$ satisfies the demand constraints on $H$.

	Fix an arbitrary vertex $v\in V$. Since every discarded edge of $G'$ in the sparsification step of the algorithm  has disjoint lists, every neighbor of $v$ which contributes to  $\free(v)$ is colored differently than it by $\chi$.   
	
		 \begin{align*}
			 |\{u\in N_{G'}(v) : \chi(u)\neq \chi(v)\}|&\geq \free(v) + \dem(v)\\
			 &\geq \free(v) + \max\{0,\min\{k,\deg_{G'}(v)\} -\free(v)\} \geq \min\{k,\deg_{G'}(v)\}	
		 \end{align*}
	Therefore, the constraints of $k$-partial coloring is satisfied by $\chi$ on $G'$ and by Lemma~\ref{lemma:witness-implies-partial}, $\chi$ is a $k$-partial coloring of $G$.
\end{proof}
\subsubsection{Space Guarantee}\label{subsec:space-guarantee}

\begin{lemma}\label{lemma:degenerate}
Let $G'=(V,E')$ be the graph induced by the edges that pass the filter of
Algorithm~\ref{alg:streaming-partial-coloring}. Then $G'$ is $k$-degenerate.
\end{lemma}

\begin{proof}
We use the equivalent characterization: a graph is $k$-degenerate if and only if
every non-empty subgraph has a vertex of degree at most $k$.

Let $S\subseteq V$ be any non-empty set of vertices and consider the induced subgraph
$G'[S]$. Let $t$ be the earliest time (in the stream) at which an edge with both
endpoints in $S$ is accepted \emph{after} every vertex in $S$ has already reached
degree at least $k$ in $G'[S]$. We claim that such a time $t$ cannot exist.

Indeed, fix any time during the stream and suppose (for contradiction) that at this
time every vertex $x\in S$ has current degree at least $k$ \emph{within} $G'[S]$.
Then, in particular, each endpoint of any arriving edge $\{u,v\}$ with $u,v\in S$
has current degree at least $k$ in the full graph $G'$ as well (since
$\deg_{G'}(x)\ge \deg_{G'[S]}(x)$ for all $x\in S$ at all times).
Therefore the filter condition ``$\deg_{G'}(u)<k$ or $\deg_{G'}(v)<k$'' fails, and
no further edge with both endpoints in $S$ can ever be accepted into $G'$.
This contradicts the definition of $t$.

Hence, at the end of the stream it cannot be the case that all vertices in $S$ have
degree at least $k+1$ in $G'[S]$. Equivalently, $G'[S]$ contains a vertex of degree
at most $k$. Since $S$ was arbitrary, every non-empty induced subgraph of $G'$ has
a vertex of degree at most $k$, and therefore $G'$ is $k$-degenerate.
\end{proof}



\begin{lemma}\label{lem:H-size}
Let $s=C\log^2 n$. With high probability,
\[
|E(H)| = O(n\log^4 n).
\]
\end{lemma}

\begin{proof}
	Since the graph $G'$ is $k$-degenerate (by Lemma~\ref{lemma:degenerate}), fix a $k$-degenerate ordering of $G'$ and define acyclic orientation of the edges such that  every vertex $v$ satisfies $\outdeg_{G'}(v)\le k$.

	Fix a vertex $v$. For each out-neighbor $u$ of $v$ in $G'$, the edge $\{v,u\}$ is stored in $H$
iff $L(v)\cap L(u)\neq \emptyset$. Conditioned on $L(v)$, these events are independent across
different out-neighbors $u$, and
\[
\Pr[L(v)\cap L(u)\neq \emptyset]\le \frac{s^2}{k+1}.
\]
Hence $\outdeg_H(v)$ is stochastically dominated by a binomial random variable with $k$ trials and mean at most $s^2=O(\log^4 n)$. A Chernoff bound implies that $\outdeg_H(v)=O(\log^4 n)$ with high probability.
A union bound over all $v\in V$ shows that with high probability, every vertex $v$ has $\outdeg_H(v)=O(\log^4 n)$ in this fixed orientation. Therefore $H$ is $O(\log^4 n)$-degenerate  and hence
$|E(H)|\le n\cdot O(\log^4 n)$ with high probability.
\end{proof}




\section{Lower Bound: A Super-Linear Space Barrier}\label{sec:lower-bound}

We establish that verifying a $k$-partial $k$-coloring requires super-linear space in the semi-streaming model. This stands in sharp contrast to $(k+1)$-coloring, which we have shown is tractable with $\tilde{O}(n)$ space. Specifically, we prove the following theorem.

\begin{theorem}\label{thm:lower_bound}
    For every constant $\varepsilon \geq 0$, let $G$ be an $n$-vertex graph presented as an adversarial insertion-only stream, and let the partial coloring parameter be $k = \Theta(n^{\frac{1}{3+\varepsilon}})$.
    Any one-pass randomized streaming algorithm $\mathcal{A}$ that solves $k$-partial $k$-coloring on $G$ with probability at least $2/3$ requires $\Omega\left(n^{1 + \frac{1}{3+\varepsilon}}\right)$
    bits of memory. In particular, for $\varepsilon=0$, the algorithm requires $\Omega(n^{4/3})$ bits.
\end{theorem}

ur proof is a reduction from the one-way \textsc{Index} communication problem. Before proceeding with the proof, we construct two ``gadgets" that force equality or inequality of colors between two vertices using the constraints of $k$-partial coloring and  will be used in the construction of lower bound graphs for the reduction.

\paragraph{Gadgets.}
For a graph $G$ and two distinct vertices $x,y\in V(G)$, 

\begin{itemize}
\item \textbf{Edge-gadget $e_{x,y}$.} Introduce a $k$-clique $K=\{p_1,\dots,p_k\}$; connect $x$ to $p_1,\dots,p_{k-1}$ and connect $y$ to $p_k$.
\item \textbf{Color-repeater $c_{x,y}$.} Introduce a $(k\!-\!1)$-clique $Q=\{q_1,\dots,q_{k-1}\}$ and connect both $x$ and $y$ to all vertices of $Q$.
\end{itemize}

\begin{lemma}[Edge-gadget forces inequality~\cite{BDGJ25}]\label{lemma:edge-gadget}
In any $k$-partial $k$-coloring $\chi$ of $G\cup e_{x,y}$, we have $\chi(x)\neq \chi(y)$.
\end{lemma}

\begin{proof}
Consider the clique $K = \{p_1, \dots, p_k\}$ within the edge gadget. Every vertex $p \in K$ has degree exactly $k$ in $G \cup e_{x,y}$. Therefore, by the definition of $k$-partial coloring, $\chi$ is a proper coloring on $e_{x,y}$. Since $K$ induces a $k$-clique and the palette has size exactly $k$, the vertices of $K$ must utilize all $k$ distinct colors. Consequently, the set of colors assigned to the subset $\{p_1, \dots, p_{k-1}\}$ consists of exactly $k-1$ distinct colors. 
    
	As the vertex $x$ is adjacent to every vertex in $\{p_1, \dots, p_{k-1}\}$, the coloring $\chi(x)$ must be the unique remaining color in the palette, which is precisely $\chi(p_k)$. Since $y$ is adjacent to $p_k$, we must have $\chi(y) \neq \chi(p_k)$ and therefore, $\chi(x)\neq \chi(y)$
\end{proof}

\begin{lemma}[Color-repeater forces equality]\label{lemma:color-repeater}
In any $k$-partial $k$-coloring $\chi$ of $G\cup c_{x,y}$, we have $\chi(x)=\chi(y)$.
\end{lemma}
\begin{proof}
	As the same arguments in lemma~\ref{lemma:edge-gadget}, the coloring $\chi$ is a proper on $c_{x,y}$ and the $(k-1)$-clique $Q$ utilizes $k-1$ distinct colors and therefore, vertices $x$ and $y$ have only one possible color that can be assigned to them and hence, $\chi(x) = \chi(y)$. 
\end{proof}

\subsection{The Reduction from \textsc{Index}}

Let Alice hold a boolean string $X \in \{0,1\}^N$ and Bob hold an index $j \in [N]$. We interpret $X$ as a matrix $A$ of dimensions $(k-1) \times \ell$, where $N = (k-1)\ell$. Bob's index corresponds to a specific entry $(g, h)$ in this matrix.

Alice and Bob create graph $G = (V(G),E_A\sqcup E_B)$ based on $A$ and $(g,h)$ respectively. We now describe the graph $G$ as follows.

\paragraph{Vertex Set.}

The vertex set $V(G)$ consists of four disjoint sets of vertices  $U = \{u_1,\ldots,u_{k-1}\}$, $V = \{v_1,\ldots,v_{k-1}\},$ $W = \{w_1,\ldots, w_{\ell}\}$ and $X$ which is a set of auxiliary vertices of cardinality $O(k^3)$ which would be used to create edge-gadgets or repeater gadgets. 

\smallskip
\noindent\textbf{Alice’s stream.}
\begin{enumerate}
	\item For each distinct pair of vertices $u_i$ and $u_{i'}$ with $i\neq i'$, insert the edge-gadget $e_{u_i,u_{i'}}$.
        \item For each distinct pair of vertices $v_i$ and $v_{i'}$ with $i\neq i'$, insert the edge-gadget $e_{v_i,v_{i'}}$.
	\item For every $i\in[k-1]$ and $h\in[\ell]$:
		If $A[i,h] = 1$, insert the edge  $\{u_i,w_h\}$, else insert the edge $\{v_i,w_h\}$. 
\end{enumerate}

\smallskip
\noindent\textbf{Bob’s stream.} For every $t\in[k-1]\setminus\{g\}$, insert the color-repeater $c_{u_t,v_t}$ and  insert the edge-gadget $e_{u_g,v_g}$.

Alice uses $O({k\choose 2})$ many edge-gadgets (for connecting each pair of vertices in $U$ and $V$ respectively) and Bob uses $(k-2)$ many color-repeaters and $1$ edge-gadget for his stream. They require $O(k^3)$ vertices to build these gadgets and utilize the vertices of $X$ to do so (which has sufficiently many vertices).  

\begin{lemma}\label{lem:structure}
Let $\chi$ be any $k$-partial $k$-coloring of the graph $G$ constructed above. Then the following properties hold:
\begin{enumerate}
    \item The vertices $u_1,\dots,u_{k-1}$ receive pairwise distinct colors, and the vertices $v_1,\dots,v_{k-1}$ receive pairwise distinct colors.
    \item For every $t\in[k-1]\setminus\{g\}$, we have $\chi(u_t)=\chi(v_t)$, while $\chi(u_g)\neq\chi(v_g)$.
    \item Let $\alpha$ denote the unique color in the palette $[k]$ that does not appear among $\{\chi(u_1),\dots,\chi(u_{k-1})\}$. Then $\chi(v_g)=\alpha$.
\end{enumerate}
\end{lemma}

\begin{proof}
The edge-gadgets inserted among the vertices of $U=\{u_1,\dots,u_{k-1}\}$ enforce $\chi(u_i)\neq\chi(u_{i'})$ for all $i\neq i'$ by Lemma~\ref{lemma:edge-gadget}. The same argument applies to the vertices of $V=\{v_1,\dots,v_{k-1}\}$, proving~(1).

For every $t\neq g$, the color-repeater $c_{u_t,v_t}$ enforces $\chi(u_t)=\chi(v_t)$ by Lemma~\ref{lemma:color-repeater}$,$ while the edge-gadget $e_{u_g,v_g}$ enforces $\chi(u_g)\neq\chi(v_g)$. This proves~(2).

Since $u_1,\dots,u_{k-1}$ receive pairwise distinct colors and the palette has size $k$, there is a unique color $\alpha$ not appearing among $\{\chi(u_1),\dots,\chi(u_{k-1})\}$. For every $t\neq g$, we have $\chi(v_t)=\chi(u_t)$ by part~(2), and hence the vertices $\{v_t : t\neq g\}$ use exactly the colors $\{\chi(u_t) : t\neq g\}$. As the vertices of $V$ must all receive distinct colors, the color $\chi(v_g)$ must differ from $\chi(v_t)$ for every $t\neq g$, and therefore must differ from all colors in $\{\chi(u_t) : t\neq g\}$. Moreover, since $\chi(u_g)\neq\chi(v_g)$ by part~(2), the color $\chi(v_g)$ must also differ from $\chi(u_g)$. It follows that $\chi(v_g)$ differs from every color in $\{\chi(u_1),\dots,\chi(u_{k-1})\}$ and  the only possible choice is the unique remaining color $\alpha$, and hence $\chi(v_g)=\alpha$ proving~(3).
\end{proof}

\begin{lemma}[Forced color at $w_h$]\label{lem:forced-color}
Let $\chi$ be any $k$-partial $k$-coloring of $G$, and let $\alpha$ be the unique
color in $[k]$ not appearing among $\{\chi(u_1),\dots,\chi(u_{k-1})\}$.
Then
\[
A[g,h]=1 \;\Longleftrightarrow\; \chi(w_h)=\chi(v_g)=\alpha.
\]
\end{lemma}
\begin{proof}
By Lemma~\ref{lem:structure}, the vertices $u_1,\dots,u_{k-1}$ receive pairwise
distinct colors, and $\chi(v_g)=\alpha$, where $\alpha$ is the unique color in
$[k]$ not appearing among $\{\chi(u_1),\dots,\chi(u_{k-1})\}$. In particular,
$\alpha\neq \chi(u_g)$.

By construction, $\deg(w_h)=k-1\le k$, hence $\chi$ is proper on all edges
incident to $w_h$. For each $i\neq g$, the vertex $w_h$ is adjacent to exactly one
of $u_i$ or $v_i$, and Lemma~\ref{lem:structure}(2) gives $\chi(u_i)=\chi(v_i)$.
Therefore the neighbors of $w_h$ contributed by rows $i\neq g$ realize exactly the
$k-2$ distinct colors $\{\chi(u_i): i\neq g\}$.

If $A[g,h]=1$, then $\{u_g,w_h\}\in E$, so the neighbors of $w_h$ realize all
colors in $\{\chi(u_1),\dots,\chi(u_{k-1})\}$. Hence the unique color in $[k]$
absent from the neighborhood of $w_h$ is $\alpha$, and therefore
$\chi(w_h)=\alpha=\chi(v_g)$.

If $A[g,h]=0$, then $\{v_g,w_h\}\in E$, so the neighbors of $w_h$ realize the set
$\{\chi(u_i): i\neq g\}\cup\{\alpha\}$. The unique missing color is then
$\chi(u_g)$, and hence $\chi(w_h)=\chi(u_g)\neq \chi(v_g)$.

Combining the two cases yields
$A[g,h]=1 \Longleftrightarrow \chi(w_h)=\chi(v_g)$.
\end{proof}

\begin{proof}[Proof of Theorem~\ref{thm:lower_bound}]
Assume there exists a one-pass randomized streaming algorithm $\mathcal A$ that,
under the promise that the input graph is $k$-partial $k$-colorable, outputs such
a coloring using $S$ bits of memory with probability at least $2/3$.

Alice simulates $\mathcal A$ on her portion of the stream and sends the memory
state to Bob. Bob continues the simulation on his portion and obtains a coloring
$\chi$.

By Lemma~\ref{lem:forced-color}, Bob can recover the bit $A[g,h]$ by checking
whether $\chi(w_h)=\chi(v_g)$ or $\chi(w_h)=\chi(u_g)$. Hence this yields a
one-way randomized protocol for \textsc{Index} on $N=(k-1)\ell$ bits using
$S$ bits of communication and success probability at least $2/3$.
Therefore $S=\Omega(N)$.

Let $\ell=\Theta(k^{3+\varepsilon})$. Then
\[
n = |U|+|V|+|W|+O(k^3)=\Theta(k^{3+\varepsilon}),
\qquad
N=(k-1)\ell=\Theta(k^{4+\varepsilon}).
\]
Thus $k=\Theta(n^{1/(3+\varepsilon)})$ and
\[
S=\Omega(N)=\Omega\!\left(n^{1+\frac{1}{3+\varepsilon}}\right).
\]
For $\varepsilon=0$, this gives the bound $\Omega(n^{4/3})$.

\end{proof}

\section*{Acknowledgements}
Avinandan would like to thank Pierre Fraigniaud and Adi Ros\'en for proposing the problem of $k$-partial $(k+1)$-coloring and many  discussions. This work was supported in part by the Research Council of Finland, Grants 363558 and 359104.

\bibliographystyle{abbrv}
\bibliography{references}

@book{MU17,
  author    = {Mitzenmacher, Michael and Upfal, Eli},
  title     = {Probability and Computing: Randomized Algorithms and Probabilistic Analysis},
  edition   = {2},
  publisher = {Cambridge University Press},
  year      = {2017}
}

@inproceedings{BCG20,
  author       = {Suman K. Bera and
                  Amit Chakrabarti and
                  Prantar Ghosh},
  editor       = {Artur Czumaj and
                  Anuj Dawar and
                  Emanuela Merelli},
  title        = {Graph Coloring via Degeneracy in Streaming and Other Space-Conscious
                  Models},
  booktitle    = {47th International Colloquium on Automata, Languages, and Programming,
                  {ICALP} 2020, Saarbr{\"{u}}cken, Germany (Virtual Conference),
                  July 8-11, 2020},
  series       = {LIPIcs},
  volume       = {168},
  pages        = {11:1--11:21},
  publisher    = {Schloss Dagstuhl - Leibniz-Zentrum f{\"{u}}r Informatik},
  year         = {2020},
  url          = {https://doi.org/10.4230/LIPIcs.ICALP.2020.11},
  doi          = {10.4230/LIPICS.ICALP.2020.11},
  bibsource    = {dblp computer science bibliography, https://dblp.org}
}

@inproceedings{HKNT22,
  author       = {Magn{\'{u}}s M. Halld{\'{o}}rsson and
                  Fabian Kuhn and
                  Alexandre Nolin and
                  Tigran Tonoyan},
  editor       = {Stefano Leonardi and
                  Anupam Gupta},
  title        = {Near-optimal distributed degree+1 coloring},
  booktitle    = {{STOC} '22: 54th Annual {ACM} {SIGACT} Symposium on Theory of Computing,
                  Rome, Italy, June 20 - 24, 2022},
  pages        = {450--463},
  publisher    = {{ACM}},
  year         = {2022},
  url          = {https://doi.org/10.1145/3519935.3520023},
  doi          = {10.1145/3519935.3520023},
  bibsource    = {dblp computer science bibliography, https://dblp.org}
}

@article{BDGJ25,
  author       = {Jan Bok and
                  Avinandan Das and
                  Anna Gujgiczer and
                  Nikola Jedlickov{\'{a}}},
  title        = {Generalizing Brooks' theorem via Partial Coloring is Hard Classically
                  and Locally},
  journal      = {CoRR},
  volume       = {abs/2508.16308},
  year         = {2025},
  url          = {https://doi.org/10.48550/arXiv.2508.16308},
  doi          = {10.48550/ARXIV.2508.16308},
  eprinttype    = {arXiv},
  eprint       = {2508.16308},
  bibsource    = {dblp computer science bibliography, https://dblp.org}
}

@article{JKS08,
  author       = {T. S. Jayram and
                  Ravi Kumar and
                  D. Sivakumar},
  title        = {The One-Way Communication Complexity of Hamming Distance},
  journal      = {Theory Comput.},
  volume       = {4},
  number       = {1},
  pages        = {129--135},
  year         = {2008},
  url          = {https://doi.org/10.4086/toc.2008.v004a006},
  doi          = {10.4086/TOC.2008.V004A006},
  bibsource    = {dblp computer science bibliography, https://dblp.org}
}

@inproceedings{DFR23,
  author       = {Avinandan Das and
                  Pierre Fraigniaud and
                  Adi Ros{\'{e}}n},
  editor       = {Alysson Bessani and
                  Xavier D{\'{e}}fago and
                  Junya Nakamura and
                  Koichi Wada and
                  Yukiko Yamauchi},
  title        = {Distributed Partial Coloring via Gradual Rounding},
  booktitle    = {27th International Conference on Principles of Distributed Systems,
                  {OPODIS} 2023, Tokyo, Japan, December 6-8, 2023},
  series       = {LIPIcs},
  volume       = {286},
  pages        = {30:1--30:22},
  publisher    = {Schloss Dagstuhl - Leibniz-Zentrum f{\"{u}}r Informatik},
  year         = {2023},
  url          = {https://doi.org/10.4230/LIPIcs.OPODIS.2023.30},
  doi          = {10.4230/LIPICS.OPODIS.2023.30},
  bibsource    = {dblp computer science bibliography, https://dblp.org}
}

@inproceedings{ACK19,
  author       = {Sepehr Assadi and
                  Yu Chen and
                  Sanjeev Khanna},
  editor       = {Timothy M. Chan},
  title        = {Sublinear Algorithms for ($\Delta$ + 1) Vertex Coloring},
  booktitle    = {Proceedings of the Thirtieth Annual {ACM-SIAM} Symposium on Discrete
                  Algorithms, {SODA} 2019, San Diego, California, USA, January 6-9,
                  2019},
  pages        = {767--786},
  publisher    = {{SIAM}},
  year         = {2019},
  url          = {https://doi.org/10.1137/1.9781611975482.48},
  doi          = {10.1137/1.9781611975482.48},
  bibsource    = {dblp computer science bibliography, https://dblp.org}
}

@article{AKM23,
  author       = {Sepehr Assadi and
                  Pankaj Kumar and
                  Parth Mittal},
  title        = {Brooks' Theorem in Graph Streams: {A} Single-Pass Semi-Streaming Algorithm
                  for $\Delta$-Coloring},
  journal      = {TheoretiCS},
  volume       = {2},
  year         = {2023},
  url          = {https://doi.org/10.46298/theoretics.23.9},
  doi          = {10.46298/THEORETICS.23.9},
  bibsource    = {dblp computer science bibliography, https://dblp.org}
}

@article{CR15,
  author       = {Daniel W. Cranston and
                  Landon Rabern},
  title        = {Brooks' Theorem and Beyond},
  journal      = {J. Graph Theory},
  volume       = {80},
  number       = {3},
  pages        = {199--225},
  year         = {2015},
  url          = {https://doi.org/10.1002/jgt.21847},
  doi          = {10.1002/JGT.21847},
  bibsource    = {dblp computer science bibliography, https://dblp.org}
}

@article{H81,
  author       = {Ian Holyer},
  title        = {The NP-Completeness of Edge-Coloring},
  journal      = {{SIAM} J. Comput.},
  volume       = {10},
  number       = {4},
  pages        = {718--720},
  year         = {1981},
  url          = {https://doi.org/10.1137/0210055},
  doi          = {10.1137/0210055},
  bibsource    = {dblp computer science bibliography, https://dblp.org}
}

@article{FKMSZ05,
  author       = {Joan Feigenbaum and
                  Sampath Kannan and
                  Andrew McGregor and
                  Siddharth Suri and
                  Jian Zhang},
  title        = {On graph problems in a semi-streaming model},
  journal      = {Theor. Comput. Sci.},
  volume       = {348},
  number       = {2-3},
  pages        = {207--216},
  year         = {2005},
  url          = {https://doi.org/10.1016/j.tcs.2005.09.013},
  doi          = {10.1016/J.TCS.2005.09.013},
  bibsource    = {dblp computer science bibliography, https://dblp.org}
}

@article{M14,
  author       = {Andrew McGregor},
  title        = {Graph stream algorithms: a survey},
  journal      = {{SIGMOD} Rec.},
  volume       = {43},
  number       = {1},
  pages        = {9--20},
  year         = {2014},
  url          = {https://doi.org/10.1145/2627692.2627694},
  doi          = {10.1145/2627692.2627694},
  bibsource    = {dblp computer science bibliography, https://dblp.org}
}

@inproceedings{AA20,
  author       = {Noga Alon and
                  Sepehr Assadi},
  editor       = {Jaroslaw Byrka and
                  Raghu Meka},
  title        = {Palette Sparsification Beyond $(\Delta+1)$ Vertex Coloring},
  booktitle    = {Approximation, Randomization, and Combinatorial Optimization. Algorithms
                  and Techniques, {APPROX/RANDOM} 2020, Virtual Conference, August 17-19,
                  2020},
  series       = {LIPIcs},
  volume       = {176},
  pages        = {6:1--6:22},
  publisher    = {Schloss Dagstuhl - Leibniz-Zentrum f{\"{u}}r Informatik},
  year         = {2020},
  url          = {https://doi.org/10.4230/LIPIcs.APPROX/RANDOM.2020.6},
  doi          = {10.4230/LIPICS.APPROX/RANDOM.2020.6},
  bibsource    = {dblp computer science bibliography, https://dblp.org}
}

@inproceedings{AY25,
  author       = {Sepehr Assadi and
                  Helia Yazdanyar},
  editor       = {Ioana Oriana Bercea and
                  Rasmus Pagh},
  title        = {Simple Sublinear Algorithms for $(\Delta + 1)$ Vertex Coloring
                  via Asymmetric Palette Sparsification},
  booktitle    = {2025 Symposium on Simplicity in Algorithms, {SOSA} 2025, New Orleans,
                  LA, USA, January 13-15, 2025},
  pages        = {1--8},
  publisher    = {{SIAM}},
  year         = {2025},
  url          = {https://doi.org/10.1137/1.9781611978315.1},
  doi          = {10.1137/1.9781611978315.1},
  bibsource    = {dblp computer science bibliography, https://dblp.org}
}

@inproceedings{NS93,
  author       = {Moni Naor and
                  Larry J. Stockmeyer},
  editor       = {S. Rao Kosaraju and
                  David S. Johnson and
                  Alok Aggarwal},
  title        = {What can be computed locally?},
  booktitle    = {Proceedings of the Twenty-Fifth Annual {ACM} Symposium on Theory of
                  Computing, May 16-18, 1993, San Diego, CA, {USA}},
  pages        = {184--193},
  publisher    = {{ACM}},
  year         = {1993},
  url          = {https://doi.org/10.1145/167088.167149},
  doi          = {10.1145/167088.167149},
  bibsource    = {dblp computer science bibliography, https://dblp.org}
}

@inproceedings{BHLOS19,
  author       = {Alkida Balliu and
                  Juho Hirvonen and
                  Christoph Lenzen and
                  Dennis Olivetti and
                  Jukka Suomela},
  editor       = {Keren Censor{-}Hillel and
                  Michele Flammini},
  title        = {Locality of Not-so-Weak Coloring},
  booktitle    = {Structural Information and Communication Complexity - 26th International
                  Colloquium, {SIROCCO} 2019, L'Aquila, Italy, July 1-4, 2019, Proceedings},
  series       = {Lecture Notes in Computer Science},
  volume       = {11639},
  pages        = {37--51},
  publisher    = {Springer},
  year         = {2019},
  url          = {https://doi.org/10.1007/978-3-030-24922-9\_3},
  doi          = {10.1007/978-3-030-24922-9\_3},
  bibsource    = {dblp computer science bibliography, https://dblp.org}
}

\end{document}